\documentclass[12pt]{amsart}

\usepackage{amssymb,amsmath,amsfonts,mathrsfs,mathtools}
\usepackage{microtype}
\usepackage[hidelinks]{hyperref}
\usepackage{mathrsfs}

\allowdisplaybreaks

\numberwithin{equation}{section}

\newtheorem{thm}{Theorem}[section]
\newtheorem{lem}[thm]{Lemma}
\newtheorem{cor}[thm]{Corollary}
\newtheorem{prop}[thm]{Proposition}
\newtheorem{defin}[thm]{Definition}
\newtheorem{rem}[thm]{Remark}

\newcommand{\A}{\mathcal A}
\newcommand{\B}{\mathcal B}
\newcommand{\C}{\mathcal C}
\newcommand{\Dalg}{\mathcal D}

\newcommand{\M}{\mathcal M}

\newcommand{\R}{\mathcal R}

\newcommand{\X}{\mathcal X}
\newcommand{\one}{\mathbf 1}

\newcommand{\cyl}[1]{[#1]}

\def\bc{{\mathbb C}}

\def\bz{{\mathbb Z}}

\def\sR{{\mathscr R}}
\def\sX{{\mathscr X}}

\def\sY{{\mathscr Y}}

\begin{document}

\title[Quasi-Product States and Factor Types]{Quasi-Product States and Factor Types for the One-Dimensional Hard-Core Model}

\author{Farrukh Mukhamedov}
\address{Farrukh Mukhamedov\\
Department of Mathematical Sciences\\
United Arab Emirates University, Al Ain, United Arab Emirates}
\email{far75m@yandex.ru, farrukh.m@uaeu.ac.ae}

\author{Yuri Suhov}
\address{Yuri Suhov\\
Statistical Laboratory, DPMMS\\
University of Cambridge\\
Cambridge CB3 0WB, UK; and Department of Mathematics\\
Pennsylvania State University\\
State College, PA 16802, USA}
\email{yms@statslab.cam.ac.uk, ims14@psu.edu}

\begin{abstract}
We study the von Neumann algebra generated by a quasi-product state associated with the one-dimensional hard-core lattice gas. If \(\kappa>0\) denotes the activity (fugacity), the transfer-matrix construction gives the stationary transition matrix
\[
P=\begin{pmatrix}p&q\\ 1&0\end{pmatrix},
\qquad p+q=1,
\qquad \kappa=\frac{q}{p^2}.
\]
The stationary Markov measure on the hard-core path space induces a faithful diagonal state on the corresponding path \(AF\)-algebra. We give a state-preserving identification of its GNS von Neumann algebra with the Feldman--Moore algebra of the measured tail equivalence relation. The relation is hyperfinite and ergodic, and its nonzero ratio set is \(\kappa^{\mathbb Z}\). Consequently, the generated algebra is the hyperfinite factor of type \(\mathrm{II}_1\) when \(\kappa=1\), and of type \(\mathrm{III}_{\lambda}\), where \(\lambda=\min\{\kappa,\kappa^{-1}\}\), when \(\kappa\ne1\). The point \(\kappa=1\) is not a thermodynamic phase transition: the pressure and particle density are analytic for all \(\kappa>0\). We also describe the precise support-corner relation with the GNS algebra of the diagonal Gibbs state on the quasi-local spin algebra, and we determine the centralizer and the flow of weights.
\end{abstract}


\maketitle

\section{Introduction}\label{sec:introduction}

Quantum Markov states provide a natural meeting point for operator algebras, probability and statistical mechanics. Since the pioneering work of Accardi and collaborators, noncommutative analogs of classical Markov chains have been studied by means of their quasi-local $C^*$-algebras, modular structure and applications to quantum statistical mechanics, see for example \cite{Ac,AFi,AFr,ALi,ALV,BKJW,FNW}. A basic structural problem is to determine the type of the von Neumann algebra generated by the GNS representation of a given state.

Powers' examples of product states already show that translation invariance and absence of interactions do not force a unique type: nontracial product states on UHF algebras generate factors of type \(\mathrm{III}_{\lambda}\) for suitable \(0<\lambda<1\) \cite{P}. The product states considered can be seen as free quantum spin lattice systems with two local states at each site with probabilities $p$ and $q=1-p$. More general studies of product states and their modular theory are in \cite{BR2}. Physically, these examples correspond to free systems, i.e. systems without a genuine interaction between neighboring sites. So there is no general dichotomy that product systems are type $\mathrm{III}_{\lambda}$ and interacting systems have to be type $\mathrm{III}_1$. The natural invariant is the multiplicative group of asymptotically realizable Radon-Nikodym quotients. For nonsingular dynamics and Markov measures, this principle is made explicit in the work of Krieger \cite{Kr1,Kr2}, and Evans' quasi-product construction provides a natural $C^*$-algebraic realization on path algebras \cite{E,O}. This paper combines these two viewpoints for the most simple exclusion interaction.

The commentary of Blekher \cite{Bl} states that a non-trivial interaction should push the associated factor to type $\mathrm{III}_1$. One of the main messages of the present paper is that this intuition may be too restrictive. Interestingly, for the hard-core model considered here the factor is of type $\mathrm{III}_{\lambda}$ for some concrete $\lambda\in(0,1)$ except for one exceptional case where the factor is of type $\mathrm{II}_1$. So, a non-trivial interaction is not enough to force type $\mathrm{III}_1$. However, this problem needs further investigations for more detailed conclusions.

It is worth noting that the case of a locally faithful Markov state has been treated in the literature.
In Ref. \cite{FM} it was proved that non-homogeneous quantum Markov states
are diagonalizable and showed that, in the translation-invariant or periodic locally faithful setting, the type of the generated von Neumann factor is determined by the spectrum
of the fundamental two-site block of the associated Hamiltonian (see also \cite{M,Oh}). In
particular, when the underlying stochastic matrix has no zero entries, so that the corresponding diagonal lifting gives a locally faithful Markov state, the factor-type problem falls within that
framework. The hard-core model studied in the present paper lies outside this
situation: its transition matrix contains a zero entry, and the diagonal state on the full
UHF algebra  is not locally faithful. It turns out that this fact affects the type of the resulting
factor.

\subsection*{The physical hard-core model and the fugacity}
Let \(\eta=(\eta_i)_{i\in\mathbb Z}\), where \(\eta_i\in\{0,1\}\) records whether site \(i\) is vacant or occupied. The hard-core interaction imposes
\begin{equation}\label{eq:hardcore-rule}
\eta_i\eta_{i+1}=0,\qquad i\in\mathbb Z,
\end{equation}
so adjacent particles are forbidden. On a finite interval \(\Lambda\), a convenient grand-canonical Hamiltonian is
\begin{equation}\label{eq:physical-Hamiltonian}
H_\Lambda(\eta)=
\begin{cases}
-\mu_{\mathrm{chem}}\displaystyle\sum_{i\in\Lambda}\eta_i,
& \text{if \eqref{eq:hardcore-rule} holds in \(\Lambda\),}\\[2mm]
+\infty,&\text{otherwise.}
\end{cases}
\end{equation}
At inverse temperature \(\beta\), an admissible configuration with \(N_\Lambda(\eta)\) particles therefore has Gibbs weight proportional to
\[
\kappa^{N_\Lambda(\eta)},
\qquad
\kappa=e^{\beta\mu_{\mathrm{chem}}}>0.
\]
Thus \(\kappa\) is the activity, or fugacity, of the hard-core gas. If an occupied site has a one-particle energy \(\varepsilon\), then the same formula holds with \(\mu_{\mathrm{chem}}\) replaced by \(\mu_{\mathrm{chem}}-\varepsilon\). We refer to \cite{BR2,G} for the standard Gibbs and transfer-matrix formalism.
For fixed fugacity, the Gibbs
measure of this model is a stationary one-dimensional Markov measure; see \cite{BHW,MRS05,SR}.
The exclusion rule and the Gibbs weights are encoded by the matrices
\begin{equation}\label{eq:intro-matrices}
B=\begin{pmatrix}1&1\\ 1&0\end{pmatrix},
\qquad
T_\kappa=\begin{pmatrix}1&\kappa\\ 1&0\end{pmatrix}.
\end{equation}
Here \(B\) is the adjacency matrix of the allowed words, and \(T_\kappa\) is a transfer matrix whose Perron--Frobenius eigenvalue is
\begin{equation}\label{eq:PF-root}
r(\kappa)=\frac{1+\sqrt{1+4\kappa}}{2},
\qquad r(\kappa)^2=r(\kappa)+\kappa.
\end{equation}
With the positive right eigenvector \(h=(r(\kappa),1)^{\mathsf T}\), the standard transfer-matrix normalization
\[
P_{ij}=\frac{(T_\kappa)_{ij}h_j}{r(\kappa)h_i}
\]
gives
\begin{equation}\label{eq:P-and-fugacity}
P=\begin{pmatrix}p&q\\ 1&0\end{pmatrix},
\qquad
p=\frac1{r(\kappa)},
\qquad
q=\frac{\kappa}{r(\kappa)^2}=1-p.
\end{equation}
Consequently,
\begin{equation}\label{eq:kappa-fugacity}
\kappa=\frac{q}{p^2},
\end{equation}
which explains the parameter governing the factor classification. Physically, \(p=P(0\mid0)\) is the probability that a vacant site is followed by another vacancy, \(q=P(1\mid0)\) is the probability that it is followed by a particle, and \(P(0\mid1)=1\) expresses the exclusion constraint. The invariant distribution is
\begin{equation}\label{eq:invariant-distribution-intro}
\pi_0=\frac1{1+q},
\qquad
\pi_1=\frac{q}{1+q}.
\end{equation}
In particular, \(\pi_1\) is the particle density.

The dimensionless pressure per site and the density are
\begin{equation}\label{eq:pressure-density}
\psi(\kappa)=\log r(\kappa),
\qquad
\rho(\kappa)=\kappa\psi'(\kappa)
=\frac12\left(1-\frac1{\sqrt{1+4\kappa}}\right).
\end{equation}
Both are real analytic for every \(\kappa>0\). This observation will be important when interpreting the exceptional point \(\kappa=1\).

\subsection*{The operator-algebraic problem}
The two-sided Gibbs measure possesses a one-sided stationary Markov marginal within the compact hard-core path space. If this criterion is employed to characterize a diagonal state within the complete spin-chain UHF algebra, the state lacks faithfulness: a localized projection onto the prohibited sequence \(11\) yields an expectation of zero. The \(AF\)-algebraic structure eliminates prohibited terms at the algebraic tier. The finite-dimensional components are categorized by permissible words concluding with the identical symbol, and the Markov measure establishes a reliable quasi-product state as per Evans \cite{E}.

Let \(\A\) be this path algebra, let \(\omega\) be the quasi-product state, and put
\[
\M=\pi_\omega(\A)''.
\]
Our main result is the following.

\begin{thm}\label{thm:main-intro}
The von Neumann algebra \(\M\) is a hyperfinite factor. With \(\kappa=q/p^2\), one has
\[
\M\cong
\begin{cases}
\text{the hyperfinite factor of type \(\mathrm{II}_1\)},&\kappa=1,\\[1mm]
\text{the hyperfinite factor of type \(\mathrm{III}_\lambda\)},
&\kappa\ne1,
\end{cases}
\qquad
\lambda=\min\{\kappa,\kappa^{-1}\}.
\]
\end{thm}

The local replacement \(000\mapsto010\) inserts one particle without changing the exterior boundary condition, and its Radon--Nikodym derivative is exactly \(\kappa\). Repeated local replacements generate the ratio group \(\kappa^{\mathbb Z}\). At \(\kappa=1\), this cocycle becomes a coboundary and the measured tail relation admits an equivalent invariant probability measure, which yields type \(\mathrm{II}_1\). For \(\kappa\ne1\), the nonzero ratio set is \(\kappa^{\mathbb Z}\), which yields type \(\mathrm{III}_\lambda\).

This type change is not a typical phase transition at finite temperatures. Equation \eqref{eq:pressure-density} indicates that the pressure, density, and all finite-volume transfer-matrix formulations maintain their analyticity at \(\kappa=1\). The alteration is modular and cohomological in nature: a precise multiplicative hindrance vanishes at a specific fugacity value. In the normalization \eqref{eq:physical-Hamiltonian}, \(\kappa=1\) represents the point of zero effective chemical potential.
Recent work has linked factor type to operational aspects of bipartite pure-state entanglement such as approximate LOCC convertibility and entanglement embezzlement \cite{vLSWW25,vLSWW24}. These results do not make the point $\kappa=1$ a thermodynamic phase transition, but they do provide an operational explanation of the type shift following canonical purification. In fact, the conventional GNS representation provides a bipartite system $(\M,\M')$ in Haag duality and the cyclic separating vector produces a pure state on $\B(L^2(\M,\omega))$. For \(\kappa\ne1\), the local algebra is of type \(\mathrm{III}_\lambda\): the arbitrary-precision pure-state LOCC convertibility exhibits the type-\(\mathrm{III}\) behavior of \cite[Theorem~52]{vLSWW25}, and the quantitative embezzlement invariant of \cite[Theorems~58 and~67]{vLSWW24} determines the subtype parameter \(\lambda\). For \(\kappa=1\), the local algebra is of type \(\mathrm{II}_1\), where majorization and semifinite entanglement monotones are still available in terms of trace. The issue is physical locality: The two parties in this standard-form purification are $\M$ and $\M'$, not automatically the von Neumann algebras of two geographical regions of the original finite-temperature chain.

We also compare the path algebra with the physically more customary quasi-local spin algebra. The GNS state on the full UHF algebra is nonfaithful, but its support compression is state-preservingly isomorphic to \((\M,\omega)\). Thus the path construction is not an unrelated replacement: it is exactly the faithful support corner of the quasi-local GNS construction. For both the right-half-chain and the two-sided quasi-local GNS factors, this gives type \(\mathrm{III}_\lambda\) when \(\kappa\ne1\), whereas at \(\kappa=1\) the full factor is type \(\mathrm{II}_\infty\) and its support corner is type \(\mathrm{II}_1\).

Finally, the path-space technique is similar in spirit to Appendix A of \cite{JNPWI25} where Izumi computes factor types for boundary algebras of Levin--Wen models. In both cases one constructs a $AF$-algebra associated to pathways in a finite graph, with a differentiated state which determines multiplicative path weights. The factor type is determined by the subgroup created by local quotients of path weights. The weights of \cite{JNPWI25} are obtained using fusion data and quantum dimensions, while the present weights come from a classical hard-core Markov chain.

The paper is organized as follows. Section~\ref{sec:preliminaries} recalls the measured-equivalence-relation tools and explains the restriction property of the ratio set. Section~\ref{sec:path-algebra} constructs the path algebra and proves the state-preserving Feldman--Moore realization, including faithfulness of the normal extension. Section~\ref{sec:type} proves Theorem~\ref{thm:main-intro}. Section~\ref{sec:quasilocal} establishes the support-corner relation with the quasi-local spin algebra. Section~\ref{sec:centralizer} describes the centralizer and the flow of weights, and Section~\ref{sec:discussion} summarizes the physical and path-space interpretation.

\section{Measured equivalence relations and factor invariants}\label{sec:preliminaries}

We recall only the facts needed below. Let \((X,\mu)\) be a standard probability space and let \(\sR\subset X\times X\) be a countable Borel equivalence relation for which \(\mu\) is quasi-invariant. We orient an arrow \((x,y)\in\sR\) from \(y\) to \(x\). Its Radon--Nikodym cocycle is denoted by
\[
D_\mu(x,y)>0.
\]
Thus, if \(\theta\) is a partial Borel isomorphism whose graph is contained in \(\sR\), then
\[
\frac{d(\mu\circ\theta)}{d\mu}(y)=D_\mu(\theta(y),y)
\]
on the domain of \(\theta\).

\begin{defin}\label{def:ratio-set}
For \(t>0\), one says that \(t\) belongs to the Krieger ratio set \(r(\sR,\mu)\) if, for every measurable set \(A\subset X\) with \(\mu(A)>0\) and every \(\varepsilon>0\), there are a measurable set \(B\subset A\) of positive measure and a partial Borel isomorphism \(\theta\) with graph in \(\sR\) such that \(\theta(B)\subset A\) and
\[
\bigl|D_\mu(\theta(y),y)-t\bigr|<\varepsilon,
\qquad y\in B.
\]
The possible value \(0\) is defined by approximation with positive derivatives; only the nonzero part will be used below.
\end{defin}

If \(\sR\) is ergodic, then \(r(\sR,\mu)\setminus\{0\}\) is a closed subgroup of \((0,\infty)\) \cite{Kr1,Kr2}. Let \(W^*(\sR,\mu)\) denote the Feldman--Moore von Neumann algebra of the measured relation \cite{FeM1,FeM2}. The canonical state is obtained by restricting a groupoid kernel to the diagonal and integrating against \(\mu\). Its modular automorphism group acts by
\begin{equation}\label{eq:FM-modular-general}
\sigma_t(f)(x,y)=D_\mu(x,y)^{it}f(x,y).
\end{equation}
The algebra \(W^*(\sR,\mu)\) is a factor exactly when \(\sR\) is ergodic. It is hyperfinite when \(\sR\) is hyperfinite. In the factor case, the Feldman--Moore and Krieger results give
\begin{equation}\label{eq:S-ratio}
S\bigl(W^*(\sR,\mu)\bigr)\setminus\{0\}
=r(\sR,\mu)\setminus\{0\},
\end{equation}
where \(S(\cdot)\) is Connes' \(S\)-invariant; equivalently,
\begin{equation}\label{eq:Connes-spectrum-ratio}
\Gamma(\sigma)=\log\bigl(r(\sR,\mu)\setminus\{0\}\bigr).
\end{equation}
Here \(\Gamma(\sigma)\) denotes the Connes spectrum of the modular action, not the ordinary Arveson spectrum of this particular action. See \cite{ConnesIII,FeM2}. Hence an ergodic hyperfinite relation with nonzero ratio set \(\lambda^{\mathbb Z}\), \(0<\lambda<1\), yields the hyperfinite factor of type \(\mathrm{III}_\lambda\). If the cocycle is a coboundary and the measure is equivalent to a finite invariant measure, the associated factor is finite.

We shall use the following restriction principle in Section~\ref{sec:type}.

\begin{lem}\label{lem:restriction-ratio}
Let \(\sR\) be ergodic and let \(Y\subset X\) have positive measure. Then the saturation of \(Y\) is conull, the restricted relation \(\sR|_Y\) is ergodic, and
\begin{equation}\label{eq:restriction-ratio-general}
r(\sR|_Y,\mu|_Y)\setminus\{0\}
=r(\sR,\mu)\setminus\{0\}.
\end{equation}
\end{lem}

\begin{proof}
The saturation of \(Y\) is an \(\sR\)-invariant set of positive measure, and is therefore conull. Let \(e=\one_Y\in L^\infty(X,\mu)\subset W^*(\sR,\mu)\). The standard groupoid compression gives
\[
eW^*(\sR,\mu)e\cong W^*(\sR|_Y,\mu|_Y).
\]
Since \(W^*(\sR,\mu)\) is a factor and \(e\ne0\), the corner is again a factor; equivalently, \(\sR|_Y\) is ergodic. Connes' \(S\)-invariant is unchanged by passage to a nonzero corner of a factor. Applying \eqref{eq:S-ratio} to the algebra and to the corner proves \eqref{eq:restriction-ratio-general}. Multiplying \(\mu|_Y\) by a positive normalizing constant changes neither its cocycle nor its ratio set.
\end{proof}

We emphasize a point relevant to Section~\ref{sec:path-algebra}. A faithful state on a \(C^*\)-algebra need not extend to a faithful normal state on the von Neumann algebra generated by its GNS representation; see Takesaki's counterexamples \cite{Takesaki74}. Takesaki also identifies important sufficient hypotheses, including the KMS setting. Rather than infer faithfulness from \(C^*\)-faithfulness alone, we prove it directly from a state-preserving Feldman--Moore realization in Proposition~\ref{prop:FM-realization}.

\section{The hard-core path algebra and the measured tail relation}\label{sec:path-algebra}

\subsection{The path \texorpdfstring{\(AF\)}{AF}-algebra and the quasi-product state}
Let
\begin{equation}\label{eq:path-space}
\sX=\bigl\{x=(x_0,x_1,\ldots)\in\{0,1\}^{\mathbb N_0}:
B_{x_kx_{k+1}}=1\text{ for all }k\ge0\bigr\}.
\end{equation}
For \(n\ge0\), let
\begin{align}
W_n&=\bigl\{\xi=(\xi_0,\ldots,\xi_n)\in\{0,1\}^{n+1}:
B_{\xi_k\xi_{k+1}}=1\text{ for }0\le k<n\bigr\},\label{eq:Wn}\\
W_n(j)&=\{\xi\in W_n:\xi_n=j\},\qquad j\in\{0,1\}.\label{eq:Wnj}
\end{align}
As was noted \(W_n\) is the collection of the admissible words of length \(n+1\), while $W_n(j)$ consists of those admissible words whose last symbol is $j$. Now, we define
\begin{equation}\label{eq:An}
\A_n=\bigoplus_{j=0,1}\B\bigl(\ell^2(W_n(j))\bigr).
\end{equation}
For example, \(W_0=\{0,1\}\), \(W_1(0)=\{00,10\}\), \(W_1(1)=\{01\}\) $\Rightarrow$
$\A_1\cong M_2(\bc)\oplus\bc$, $\A_2\cong M_3(\bc)
\oplus M_2(\bc)$, and so on.
This illustrates the basic point of the construction: all noncommutative matrix units are retained, but only inside blocks indexed by admissible words.

For \(\xi,\eta\in W_n(j)\), write \(e^{(n)}_{\xi,\eta}\) for the corresponding matrix unit. The embedding map \(\iota_n:\A_n\to\A_{n+1}\) is
\begin{equation}\label{eq:embedding}
\iota_n\bigl(e^{(n)}_{\xi,\eta}\bigr)
=
\sum_{a:\,B_{\xi_n a}=1}e^{(n+1)}_{\xi a,\eta a}.
\end{equation}
Because \(\xi_n=\eta_n\), the same extensions occur on both sides. Direct verification of the matrix-unit relations shows that \(\iota_n\) is a unital injective \(*\)-homomorphism. We set
\begin{equation}\label{eq:path-AF}
\A=\varinjlim(\A_n,\iota_n).
\end{equation}

Let \(\Dalg_n\subset\A_n\) be the diagonal algebra spanned by the projections \(e^{(n)}_{\xi,\xi}\).
Given $n\in\bz_+$, define a conditional expectation $E_n:\A_n\to\Dalg_n$ by
\begin{equation}\label{eq:finite-diagonal-expectation}
E_n\bigl(e^{(n)}_{\xi,\eta}\bigr)
=\delta_{\xi,\eta}e^{(n)}_{\xi,\xi}.
\end{equation}

The expectations \(E_n\) are compatible with the embeddings and induce a faithful conditional expectation
\begin{equation}\label{eq:global-diagonal-expectation}
E:\A\longrightarrow\Dalg,
\qquad
\Dalg=\varinjlim\Dalg_n\cong C(\sX).
\end{equation}
Under this identification,
\begin{equation}\label{eq:cylinder}
e^{(n)}_{\xi,\xi}\longleftrightarrow\one_{\cyl{\xi}},
\qquad
\cyl{\xi}=\{x\in\sX:(x_0,\ldots,x_n)=\xi\}.
\end{equation}

Let \(\pi=(\pi_0,\pi_1)\) be the invariant distribution in \eqref{eq:invariant-distribution-intro}, and let \(\mu=\mu_P\) be the stationary Markov measure on \(\sX\):
\begin{equation}\label{eq:Markov-cylinder}
\mu(\cyl{\xi})
=\pi_{\xi_0}\prod_{k=0}^{n-1}P_{\xi_k\xi_{k+1}},
\qquad \xi\in W_n.
\end{equation}
Every admissible cylinder has positive measure.
Compatibility follows from
\[
\sum_{a:\,B_{\xi_n a}=1}\mu(\cyl{\xi a})
=\mu(\cyl{\xi})\sum_aP_{\xi_n a}
=\mu(\cyl{\xi}).
\]
Define
\begin{equation}\label{eq:state-on-matrix-units}
\omega_n\bigl(e^{(n)}_{\xi,\eta}\bigr)
=\delta_{\xi,\eta}\mu(\cyl{\xi}), \qquad
\xi,\eta\in W_n,\ \xi_n=\eta_n.
\end{equation}
Now, using the compatibility we may establish that the states $\{\omega_n\}$ are compatible on the matrix units. Let $\xi,\eta\in W_n$ with
$\xi_n=\eta_n$. By \eqref{eq:embedding} and \eqref{eq:state-on-matrix-units},
\begin{eqnarray*}
\omega_{n+1}\bigl(\iota_n(e^{(n)}_{\xi,\eta})\bigr)
&=&
\sum_{k:\,B(\xi_n,k)=1}
\omega_{n+1}\bigl(e^{(n+1)}_{\xi k,\eta k}\bigr)=
\delta_{\xi,\eta}
\sum_{k:\,B(\xi_n,k)=1} \mu([\xi k])\\
&=&
\delta_{\xi,\eta}\mu([\xi])\sum_{k:\,B(\xi_n,k)=1}P_{\xi_n k}=
\delta_{\xi,\eta}\mu([\xi])
=
\omega_{P,n}\bigl(e^{(n)}_{\xi,\eta}\bigr).
\end{eqnarray*}
Thus \eqref{eq:state-on-matrix-units} defines a state on \(\A\), and
\begin{equation}\label{eq:state-diagonal}
\omega=\mu\circ E.
\end{equation}
Since \(E\) is faithful and \(\mu\) is strictly positive on every nonempty clopen subset of \(\sX\), \(\omega\) is faithful as a \(C^*\)-algebraic state. It is a quasi-product state in the terminology of Evans \cite{E}.

For later reference, the restriction of \(\omega\) to \(\A_n\) has density
\begin{equation}\label{eq:local-density}
\rho_n=\sum_{\xi\in W_n}\mu(\cyl{\xi})e^{(n)}_{\xi,\xi}
\end{equation}
with respect to the canonical unnormalized trace on \(\A_n\). Hence the finite-volume diagonal Hamiltonian is
\begin{equation}\label{eq:local-Hamiltonian}
H_n=-\log\rho_n
=-\sum_{\xi\in W_n}
\left(\log\pi_{\xi_0}+\sum_{k=0}^{n-1}\log P_{\xi_k\xi_{k+1}}\right)
e^{(n)}_{\xi,\xi}.
\end{equation}
The initial factor \(\pi_{\xi_0}\) is the stationary boundary term.

\subsection{The state-preserving Feldman--Moore realization}
Define the tail equivalence relation
\begin{equation}\label{eq:tail-relation}
\sR=\bigl\{(x,y)\in\sX\times\sX:\text{there is }N\text{ such that }x_k=y_k
\text{ for all }k\ge N\bigr\}.
\end{equation}
For \(n\ge0\), put
\[
\sR_n=\{(x,y)\in\sR: \ x_k=y_k\text{ for all }k\ge n\}.
\]
Then \(\sR_0\subset\sR_1\subset\cdots\), each \(\sR_n\) has finite equivalence classes, and \(\sR=\bigcup_n\sR_n\).

If \(\xi,\eta\in W_n\) have the same terminal symbol, let
\begin{equation}\label{eq:basic-bisection}
Z(\xi,\eta)=
\bigl\{(x,y)\in\sR_n:(x_0,\ldots,x_n)=\xi,
\ (y_0,\ldots,y_n)=\eta\bigr\}.
\end{equation}
This is a compact open bisection with range \(\cyl{\xi}\) and source \(\cyl{\eta}\). The family of such bisections is a basis for the \(AF\)-groupoid topology on \(\sR\).

\begin{prop}\label{prop:FM-realization}
The following statements hold.
\begin{enumerate}
\item[\rm(i)] The assignment
\begin{equation}\label{eq:matrix-unit-groupoid-map}
e^{(n)}_{\xi,\eta}\longmapsto\one_{Z(\xi,\eta)}
\end{equation}
extends to a \(*\)-isomorphism \(\A\cong C^*(\sR)\). Under this isomorphism, \(E\) is restriction to the unit space \(\sX\).
\item[\rm(ii)] The measure \(\mu\) is quasi-invariant for \(\sR\). If \((x,y)\in\sR\) and \(x_k=y_k\) for every \(k\ge N\), then
\begin{equation}\label{eq:cocycle}
D_\mu(x,y)
=\frac{\pi_{x_0}}{\pi_{y_0}}
\prod_{k=0}^{N-1}
\frac{P_{x_kx_{k+1}}}{P_{y_ky_{k+1}}}.
\end{equation}
In particular, on \(Z(\xi,\eta)\),
\begin{equation}\label{eq:cocycle-cylinder-ratio}
D_\mu(x,y)=\frac{\mu(\cyl{\xi})}{\mu(\cyl{\eta})}.
\end{equation}
\item[\rm(iii)] There is a state-preserving isomorphism
\[
(\pi_\omega(\A)'',\omega)
\cong (W^*(\sR,\mu),\varphi_\mu),
\]
where \(\varphi_\mu\) is the canonical Feldman--Moore state. Consequently, the normal extension of \(\omega\) to \(\pi_\omega(\A)''\) is faithful.
\end{enumerate}
\end{prop}

\begin{proof}
For fixed \(n\), the characteristic functions of the bisections satisfy
\[
\one_{Z(\xi,\eta)}*\one_{Z(\zeta,\tau)}
=\delta_{\eta,\zeta}\one_{Z(\xi,\tau)},
\qquad
\one_{Z(\xi,\eta)}^*=\one_{Z(\eta,\xi)}.
\]
Refining a bisection by one edge gives
\[
Z(\xi,\eta)=\bigsqcup_{a:\,B_{\xi_n a}=1}Z(\xi a,\eta a),
\]
which is exactly the embedding formula \eqref{eq:embedding}. Thus \eqref{eq:matrix-unit-groupoid-map} identifies the algebraic inductive limit \(\bigcup_n\A_n\) with the locally constant compactly supported convolution algebra of the \(AF\)-groupoid \(\R\). Since \(AF\)-groupoids are amenable, the full and reduced groupoid norms agree, and the map extends to \(\A\cong C^*(\sR)\); see \cite{Renault}. Restriction to the diagonal sends \(\one_{Z(\xi,\eta)}\) to \(\delta_{\xi,\eta}\one_{\cyl{\xi}}\), proving (i).

Every source and range cylinder in \eqref{eq:basic-bisection} has positive \(\mu\)-measure. The conditional law of the continuation after time \(n\) depends only on the common terminal symbol \(\xi_n=\eta_n\). Hence the prefix-replacement map from \(\cyl{\eta}\) to \(\cyl{\xi}\) has constant derivative \(\mu(\cyl{\xi})/\mu(\cyl{\eta})\). This proves quasi-invariance, \eqref{eq:cocycle-cylinder-ratio}, and, by substituting \eqref{eq:Markov-cylinder}, formula \eqref{eq:cocycle}.

For completeness, we recall how the normal representation enters. Let \(\nu_l\) be the left counting measure on \(\sR\),
\[
\int_\sR F\,d\nu_l
=\int_\sX\sum_{y:\,(x,y)\in\sR}F(x,y)\,d\mu(x),
\]
and let \(L\) be the left convolution representation on \(L^2(\sR,\nu_l)\). The diagonal vector \(\Omega=\one_\Delta\), where \(\Delta=\{(x,x):x\in\sX\}\), satisfies
\begin{equation}\label{eq:canonical-vector-state}
\langle L(f)\Omega,\Omega\rangle
=\int_\sX f(x,x)\,d\mu(x).
\end{equation}
The left and right Feldman--Moore representations commute, and \(\Omega\) is cyclic for both; hence \(\Omega\) is separating for \(L(C_c(\sR))''\) \cite{FeM1,FeM2}. Therefore the vector state in \eqref{eq:canonical-vector-state} is faithful and normal on \(W^*(\sR,\mu)\).

By (i), \eqref{eq:state-diagonal}, and \eqref{eq:canonical-vector-state}, the vector state of \(\Omega\) agrees with \(\omega\) on the dense \(*\)-subalgebra spanned by the matrix units. The uniqueness of the GNS representation gives the state-preserving isomorphism in (iii). In particular, the GNS vector is separating for \(\pi_\omega(\A)''\), so the normal extension of \(\omega\) is faithful.
\end{proof}

From now on, we use Proposition~\ref{prop:FM-realization} to identify
\begin{equation}\label{eq:M-identification}
\M:=\pi_\omega(\A)''=W^*(\sR,\mu),
\end{equation}
and we use the same symbol \(\omega\) for the faithful normal state on \(\M\).

\begin{prop}\label{prop:modular-matrix-units}
For \(\xi,\eta\in W_n\) with \(\xi_n=\eta_n\),
\begin{equation}\label{eq:modular-matrix-units}
\sigma_t^\omega\bigl(e^{(n)}_{\xi,\eta}\bigr)
=\left(\frac{\mu(\cyl{\xi})}{\mu(\cyl{\eta})}\right)^{it}
e^{(n)}_{\xi,\eta},
\qquad t\in\mathbb R.
\end{equation}
Every local matrix unit is entire analytic for \(\sigma^\omega\).
\end{prop}

\begin{proof}
Under \eqref{eq:matrix-unit-groupoid-map}, the matrix unit is the characteristic function of \(Z(\xi,\eta)\), and the cocycle is constant there by \eqref{eq:cocycle-cylinder-ratio}. Formula \eqref{eq:modular-matrix-units} therefore follows immediately from the Feldman--Moore modular formula \eqref{eq:FM-modular-general}. The scalar coefficient has an entire extension in the time variable.
\end{proof}

\begin{rem}\label{rem:Tomita-proof}
Proposition~\ref{prop:modular-matrix-units} avoids a domain issue that arises in a direct computation with the unbounded Tomita operator. A verification of an adjoint identity merely on a dense subspace is insufficient unless that subspace is known to be a core. The measured-relation realization proves both faithfulness and the modular formula before any unbounded-operator calculation is required.
\end{rem}

\section{Hyperfiniteness, ergodicity, and factor type}\label{sec:type}

Throughout this section,
\begin{equation}\label{eq:kappa-lambda}
\kappa=\frac{q}{p^2},
\qquad
\lambda=\min\{\kappa,\kappa^{-1}\}.
\end{equation}

\subsection{The factor property}

\begin{prop}\label{prop:hyperfinite-factor}
The von Neumann algebra \(\M=W^*(\sR,\mu)\) is a hyperfinite factor.
\end{prop}

\begin{proof}
Since \(\sR=\bigcup_n\sR_n\) and every \(\sR_n\)-class is finite, \(\sR\) is a hyperfinite equivalence relation. Hence \(\M\) is hyperfinite.

It remains to prove ergodicity. The transition matrix is primitive, because
\[
P^2=
\begin{pmatrix}
p^2+q&pq\\ p&q
\end{pmatrix}
\]
has strictly positive entries. The stationary Markov chain is therefore strongly mixing: there are constants \(C>0\) and \(0<\alpha<1\) such that, if
\[
\mathcal F_m=\sigma(x_0,\ldots,x_m),
\qquad
\mathcal G_n=\sigma(x_n,x_{n+1},\ldots),
\]
then
\begin{equation}\label{eq:strong-mixing}
\bigl|\mu(B\cap C)-\mu(B)\mu(C)\bigr|
\le C\alpha^{n-m}
\end{equation}
for \(B\in\mathcal F_m\), \(C\in\mathcal G_n\), and \(n>m\). This is the standard finite-state Markov mixing estimate obtained from the Perron--Frobenius convergence of \(P^k\).

Let \(A\) be \(\sR\)-invariant. Membership in \(A\) is unchanged by replacing any finite prefix, so \(A\) belongs, modulo null sets, to the tail sigma-algebra \(\bigcap_n\mathcal G_n\). Given \(\varepsilon>0\), choose \(m\) and \(B\in\mathcal F_m\) such that \(\mu(A\triangle B)<\varepsilon\). Since \(A\in\mathcal G_n\) for every \(n\), \eqref{eq:strong-mixing} and then \(n\to\infty\) give
\[
\bigl|\mu(A\cap B)-\mu(A)\mu(B)\bigr|=0.
\]
Consequently,
\[
|\mu(A)-\mu(A)^2|
\le \mu(A\setminus B)
+|\mu(A\cap B)-\mu(A)\mu(B)|
+\mu(A)|\mu(B)-\mu(A)|
<2\varepsilon.
\]
Letting \(\varepsilon\downarrow0\) yields \(\mu(A)\in\{0,1\}\). Thus \(\sR\) is ergodic, so \(\M\) is a factor.
\end{proof}

\subsection{Computation of the ratio set}
Let
\begin{equation}\label{eq:X0-R0}
\sX_0=\cyl{0}=\{x\in\sX:x_0=0\},
\qquad
\sR^{(0)}=\sR|_{\X_0}.
\end{equation}
The set \(\sX_0\) has measure \(\pi_0>0\) and meets every \(\sR\)-class: if \(x_0=1\), then \(x_1=0\), and the admissible path \(y=(0,x_1,x_2,\ldots)\) lies in \(\sX_0\) and is tail equivalent to \(x\). Lemma~\ref{lem:restriction-ratio} therefore gives
\begin{equation}\label{eq:restriction-ratio-hardcore}
r(\sR,\mu)\setminus\{0\}
=r(\sR^{(0)},\mu|_{\sX_0})\setminus\{0\}.
\end{equation}
This is the precise form of the restriction step used in the computation.

For an admissible word \(\xi=(0,\xi_1,\ldots,\xi_n)\) beginning at \(0\), set
\begin{equation}\label{eq:m-xi}
m(\xi)=\#\{0\le k<n:\xi_k=0,\ \xi_{k+1}=1\}.
\end{equation}

\begin{lem}\label{lem:word-ratio}
Let \(\xi\) and \(\eta\) be admissible words of the same length, both beginning at \(0\) and ending at the same symbol. Then
\begin{equation}\label{eq:word-ratio}
\frac{\mu(\cyl{\xi})}{\mu(\cyl{\eta})}
=\kappa^{m(\xi)-m(\eta)}.
\end{equation}
In particular,
\begin{equation}\label{eq:local-kappa-ratio}
\frac{\mu(\cyl{010})}{\mu(\cyl{000})}=\kappa.
\end{equation}
\end{lem}

\begin{proof}
Suppose first that \(\xi\) ends at \(0\). The number of \(10\)-transitions equals \(m(\xi)\), and the number of \(00\)-transitions is \(n-2m(\xi)\). Hence
\[
\mu(\cyl{\xi})
=\pi_0p^{n-2m(\xi)}q^{m(\xi)}
=\pi_0p^n\kappa^{m(\xi)}.
\]
If \(\xi\) ends at \(1\), the number of \(10\)-transitions is \(m(\xi)-1\), and the number of \(00\)-transitions is \(n+1-2m(\xi)\). Thus
\[
\mu(\cyl{\xi})
=\pi_0p^{n+1-2m(\xi)}q^{m(\xi)}
=\pi_0p^{n+1}\kappa^{m(\xi)}.
\]
The common prefactor cancels in the ratio with \(\eta\), proving \eqref{eq:word-ratio}. Equation \eqref{eq:local-kappa-ratio} is the case \(\xi=010\), \(\eta=000\).
\end{proof}

\begin{lem}\label{lem:density-cylinder}
Let \(A\subset\sX_0\) be measurable with \(\mu(A)>0\), and let \(\delta>0\). There exists an admissible word \(u\), beginning and ending at \(0\), such that
\begin{equation}\label{eq:density-cylinder}
\mu(A\cap\cyl{u})>(1-\delta)\mu(\cyl{u}).
\end{equation}
\end{lem}

\begin{proof}
Let \(\mathcal H_n\) be the finite sigma-algebra generated by the length-\((n+1)\) cylinders in \(\sX_0\). By martingale convergence,
\[
\mathbb E(\one_A\mid\mathcal H_n)(x)\longrightarrow\one_A(x)
\]
for almost every \(x\in\sX_0\). Choose a density point \(x\in A\). Every hard-core path contains infinitely many zeros, so one can choose \(n\) with \(x_n=0\) and \(\mathbb E(\one_A\mid\mathcal H_n)(x)>1-\delta\). The prefix \(u=(x_0,\ldots,x_n)\) then has the required properties.
\end{proof}

\begin{prop}\label{prop:ratio-set}
The nonzero ratio set is
\begin{equation}\label{eq:ratio-set}
r(\sR,\mu)\setminus\{0\}=\kappa^{\mathbb Z}.
\end{equation}
Equivalently,
\begin{equation}\label{eq:Connes-spectrum}
\Gamma(\sigma^\omega)=(\log\kappa)\mathbb Z.
\end{equation}
\end{prop}

\begin{proof}
Let \((x,y)\in\sR^{(0)}\), and choose \(N\) such that \(x_k=y_k\) for all \(k\ge N\). The prefixes
\[
\xi=(x_0,\ldots,x_N),
\qquad
\eta=(y_0,\ldots,y_N)
\]
begin at \(0\), have the same length, and end at the common symbol \(x_N=y_N\). By \eqref{eq:cocycle-cylinder-ratio} and Lemma~\ref{lem:word-ratio},
\[
D_\mu(x,y)=\frac{\mu(\cyl{\xi})}{\mu(\cyl{\eta})}
\in\kappa^{\mathbb Z}.
\]
Since \(\kappa^{\mathbb Z}\) is closed in \((0,\infty)\),
\begin{equation}\label{eq:ratio-upper}
r(\sR^{(0)},\mu|_{\sX_0})\setminus\{0\}
\subset\kappa^{\mathbb Z}.
\end{equation}

We next prove that \(\kappa\) is an essential ratio. Let \(A\subset\sX_0\) have positive measure. Choose \(\delta>0\) such that
\begin{equation}\label{eq:delta-choice}
p^2-(1+\kappa^{-1})\delta>0,
\end{equation}
and choose \(u\) as in Lemma~\ref{lem:density-cylinder}. Define the prefix-replacement map
\begin{equation}\label{eq:Tu}
T_u:\cyl{u00}\longrightarrow\cyl{u10},
\qquad
T_u(u00z)=u10z.
\end{equation}
Both words end at \(0\), so the same tail \(z\) may be appended. The graph of \(T_u\) lies in \(\sR^{(0)}\), and
\begin{equation}\label{eq:Tu-derivative}
D_\mu(T_u(x),x)
=\frac{\mu(\cyl{u10})}{\mu(\cyl{u00})}
=\frac{q}{p^2}=\kappa
\end{equation}
throughout its domain.

Set
\[
C=A\cap\cyl{u00}\cap T_u^{-1}(A\cap\cyl{u10}).
\]
Since \(\mu(\cyl{u00})=p^2\mu(\cyl{u})\), the density estimate and \eqref{eq:Tu-derivative} give
\begin{align*}
\mu(C)
&\ge \mu(\cyl{u00})
-\mu(\cyl{u00}\setminus A)
-\mu\bigl(T_u^{-1}(\cyl{u10}\setminus A)\bigr)\\
&\ge \bigl[p^2-(1+\kappa^{-1})\delta\bigr]\mu(\cyl{u})>0.
\end{align*}
Thus \(T_u\) maps a positive-measure subset of \(A\) into \(A\) with derivative exactly \(\kappa\). Definition~\ref{def:ratio-set} gives \(\kappa\in r(\sR^{(0)},\mu|_{\sX_0})\).

By Lemma~\ref{lem:restriction-ratio}, the restricted relation \(\sR^{(0)}\) is ergodic. Its nonzero ratio set is therefore a subgroup, so it contains \(\kappa^{\mathbb Z}\). Combining this with \eqref{eq:ratio-upper} and \eqref{eq:restriction-ratio-hardcore} proves \eqref{eq:ratio-set}. Formula \eqref{eq:Connes-spectrum} follows from \eqref{eq:Connes-spectrum-ratio}; here again \(\Gamma\) is the Connes spectrum.
\end{proof}

\subsection{Classification}

\begin{thm}\label{thm:type-classification}
Let \(\M=\pi_\omega(\A)''\).
\begin{enumerate}
\item[\rm(i)] If \(\kappa\ne1\), then \(\M\) is the hyperfinite factor of type \(\mathrm{III}_\lambda\), where \(\lambda=\min\{\kappa,\kappa^{-1}\}\).
\item[\rm(ii)] If \(\kappa=1\), then \(\M\) is the hyperfinite factor of type \(\mathrm{II}_1\).
\end{enumerate}
\end{thm}

\begin{proof}
If \(\kappa\ne1\), Proposition~\ref{prop:ratio-set} gives
\[
S(\M)\setminus\{0\}=\kappa^{\mathbb Z}=\lambda^{\mathbb Z}.
\]
Together with Proposition~\ref{prop:hyperfinite-factor}, this is precisely the type \(\mathrm{III}_\lambda\) case.

Now assume \(\kappa=1\), equivalently \(q=p^2\). A direct count of the three allowed transitions gives, for every \(\xi=(\xi_0,\ldots,\xi_n)\in W_n\),
\begin{equation}\label{eq:kappa-one-cylinder}
\mu(\cyl{\xi})
=\pi_0p^{n+\delta_{\xi_0,1}+\delta_{\xi_n,1}}.
\end{equation}
Consequently, if \((x,y)\in\sR\),
\begin{equation}\label{eq:coboundary}
D_\mu(x,y)
=p^{\delta_{x_0,1}-\delta_{y_0,1}}
=\frac{h(x)}{h(y)},
\qquad
h(x)=p^{\delta_{x_0,1}}.
\end{equation}
Let \(c>0\) normalize the finite measure
\[
d\nu=c\,h^{-1}\,d\mu
\]
to a probability measure. Under an equivalent change of measure, the cocycle transforms as
\[
D_\nu(x,y)=\frac{h(y)}{h(x)}D_\mu(x,y)=1.
\]
Thus \(\nu\) is \(\sR\)-invariant. The canonical state associated with \(\nu\) is a faithful normal trace on the same measured-relation factor. Hence \(\M\) is finite. It is infinite-dimensional and hyperfinite, so it is the hyperfinite factor of type \(\mathrm{II}_1\).
\end{proof}

Theorem~\ref{thm:type-classification} is the detailed form of Theorem~\ref{thm:main-intro}. Two elementary parameter choices may be useful. If \(p=q=1/2\), then \(\kappa=2\) and \(\M\) has type \(\mathrm{III}_{1/2}\). If
\[
p=\frac{\sqrt5-1}{2},
\qquad
q=\frac{3-\sqrt5}{2},
\]
then \(q=p^2\), so \(\kappa=1\) and \(\M\) has type \(\mathrm{II}_1\).

\begin{rem}\label{rem:no-phase-transition}
The exceptional operator-algebraic point is \(\kappa=1\), but \eqref{eq:pressure-density} is analytic there. In particular,
\[
\rho(1)=\frac12\left(1-\frac1{\sqrt5}\right),
\]
and there is no singularity in the pressure or density. The type change records the cohomology class and essential value group of the tail cocycle, not coexistence of Gibbs phases or a nonanalytic free energy.
\end{rem}

\section{Relation with the quasi-local spin algebra}\label{sec:quasilocal}

The path algebra is mathematically natural because the hard-core constraint has already been imposed. We now relate it to the GNS construction on the full quasi-local algebra, where forbidden local configurations are retained as observables but have zero Gibbs probability.

Let
\begin{equation}\label{eq:UHF-plus}
\widehat\A_+=\bigotimes_{n=0}^\infty M_2(\mathbb C)
\end{equation}
be the right-half-chain UHF algebra. Its canonical diagonal is \(C(\sY)\), where
\[
\sY=\{0,1\}^{\mathbb N_0},
\]
and we write \(\widehat E_+:\widehat\A_+\to C(\sY)\) for the diagonal expectation. Regard \(\mu\) as a probability measure on \(\sY\) supported on the closed hard-core subshift \(\sX\), and define
\begin{equation}\label{eq:UHF-state}
\widehat\omega_+=\mu\circ\widehat E_+.
\end{equation}
This state is not faithful. Let
\[
\widehat\M_+=\pi_{\widehat\omega_+}(\widehat\A_+)''
\]
and let \(s_+\in\widehat\M_+\) be the support projection of the normal extension of \(\widehat\omega_+\).

For \(n\ge0\), let \(Q_n\) be the local projection onto the span of the admissible words in the first \(n+1\) tensor factors:
\begin{equation}\label{eq:Qn}
Q_n=\sum_{\xi\in W_n}|\xi\rangle\langle\xi|\otimes\one_{[n+1,\infty)}.
\end{equation}
Then \(Q_{n+1}\le Q_n\) and \(\widehat\omega_+(Q_n)=1\).

\begin{prop}\label{prop:support-corner}
There is a state-preserving isomorphism
\begin{equation}\label{eq:support-corner}
\bigl(s_+\widehat\M_+s_+,\widehat\omega_+|_{s_+\widehat\M_+s_+}\bigr)
\cong(\M,\omega).
\end{equation}
The projection \(s_+\) has central support one in \(\widehat\M_+\). Consequently, \(\widehat\M_+\) is a factor and
\begin{equation}\label{eq:UHF-factor-types}
\widehat\M_+\text{ has type }
\begin{cases}
\mathrm{III}_\lambda,&\kappa\ne1,\\
\mathrm{II}_\infty,&\kappa=1.
\end{cases}
\end{equation}
\end{prop}

\begin{proof}
Let \(\widehat\sR\) be the full tail equivalence relation on \(\sY\). The UHF algebra \(\widehat\A_+\) is the \(AF\)-groupoid algebra \(C^*(\widehat\sR)\). For \(x\in\X\), let \(\lambda_x\) be the regular representation on \(\ell^2([x]_{\widehat\sR})\). The GNS representation of the diagonal state \(\mu\circ\widehat E_+\) has the concrete form
\begin{equation}\label{eq:induced-UHF-representation}
\mathcal H_{\widehat\omega_+}
\cong\int_\sX^\oplus\ell^2([x]_{\widehat\R})\,d\mu(x),
\qquad
\pi_{\widehat\omega_+}
\cong\int_\sX^\oplus\lambda_x\,d\mu(x),
\end{equation}
with cyclic vector \(\Omega=(\delta_x)_{x\in\X}\); this is the usual induced representation of a groupoid diagonal state \cite{Renault}. Its diagonal support is precisely the saturation
\[
\widehat\sR[\sX]=\{y\in \sY:y\text{ is eventually equal to some }x\in\sX\}.
\]

Put
\[
\sX^{(n)}=\{y\in \sY:\ B_{y_ky_{k+1}}=1\text{ for }0\le k<n\}.
\]
For \(y\in[x]_{\widehat\sR}\), the diagonal projection \(\pi_{\widehat\omega_+}(Q_n)\) acts by
\[
\pi_{\widehat\omega_+}(Q_n)\delta_y
=\one_{\X^{(n)}}(y)\delta_y.
\]
Since \(\sX^{(n)}\downarrow\sX\), the projections decrease strongly to the diagonal projection \(s=\one_\sX\). Moreover, \(\widehat\omega_+(Q_n)=1\) for all \(n\), so normality gives \(\widehat\omega_+(s)=1\).

Compression by \(s=\one_\sX\) restricts both endpoints of every tail arrow to \(\sX\). Therefore
\begin{equation}\label{eq:compressed-relation}
s\widehat\M_+s\cong W^*(\widehat\sR|_\sX,\mu)
=W^*(\sR,\mu)=\M.
\end{equation}
There is also a concrete finite-word description of this compression. If \(f_{\xi,\eta}\) is the UHF matrix unit replacing a word \(\eta\) by a word \(\xi\) of the same length, then internally forbidden words vanish after compression, while for admissible \(\xi,\eta\) one has, at the next level,
\begin{equation}\label{eq:compressed-UHF-matrix-unit}
sf_{\xi,\eta}s
=\sum_{a:\,B_{\xi_n a}=B_{\eta_n a}=1}
e^{(n+1)}_{\xi a,\eta a}.
\end{equation}
Each surviving summand has a common terminal symbol and is therefore a path-algebra matrix unit. Conversely, every path matrix unit is obtained from a compressed UHF matrix unit by multiplying with the diagonal projection that fixes the next symbol. Thus the compressed local matrix units generate all path matrix units. Equation \eqref{eq:compressed-UHF-matrix-unit} is the finite-level form of \eqref{eq:compressed-relation}.

Under this identification, the compressed vector state is the canonical state \(\omega\). Proposition~\ref{prop:FM-realization} shows that this compressed state is faithful. Hence \(s\) is exactly the support projection \(s_+\), proving \eqref{eq:support-corner}.

Every point of the represented saturation \(\widehat\sR[\sX]\) is \(\widehat\sR\)-equivalent to a point of \(\sX\). Equivalently, the saturation of the diagonal projection \(\one_\sX\) is the identity of the represented groupoid algebra. Thus \(s_+\) has central support one. A von Neumann algebra with a full factor corner is itself a factor.

If \(\kappa\ne1\), a full corner and the ambient factor have the same type \(\mathrm{III}_\lambda\). Suppose now that \(\kappa=1\). The full corner is type \(\mathrm{II}_1\), so \(\widehat\M_+\) is a semifinite factor. It cannot be type I, because every nonzero corner of a type I factor is again type I. We show that it is not finite. If it were type \(\mathrm{II}_1\), let \(\tau\) be its normalized trace. The representation of the simple UHF algebra \(\widehat\A_+\) is faithful, and \(\tau\circ\pi_{\widehat\omega_+}\) is therefore the unique tracial state of \(\widehat\A_+\). Since the support projection is dominated by every projection of state value one,
\[
s_+\le\pi_{\widehat\omega_+}(Q_n),
\qquad n\ge0.
\]
On the other hand,
\begin{equation}\label{eq:trace-Qn}
\tau\bigl(\pi_{\widehat\omega_+}(Q_n)\bigr)
=\frac{|W_n|}{2^{n+1}}.
\end{equation}
The number \(|W_n|\) is the Fibonacci number \(F_{n+3}\), so the right-hand side tends to zero because the golden ratio is smaller than \(2\). It follows that \(0<\tau(s_+)\le |W_n|/2^{n+1}\to0\), a contradiction. Hence \(\widehat\M_+\) is type \(\mathrm{II}_\infty\).
\end{proof}

\begin{prop}\label{prop:two-sided-UHF}
Let
\[
\sY_{\mathbb Z}=\{0,1\}^{\mathbb Z},
\qquad
\sX_{\mathbb Z}=\{x\in \sY_{\mathbb Z}:x_ix_{i+1}=0\text{ for all }i\in\mathbb Z\},
\]
let \(\mu_{\mathbb Z}\) be the stationary two-sided Markov measure with transition matrix \(P\), and put
\[
\widehat\A_{\mathbb Z}=\bigotimes_{i\in\mathbb Z}M_2(\mathbb C),
\qquad
\widehat\omega_{\mathbb Z}=\mu_{\mathbb Z}\circ\widehat E_{\mathbb Z},
\qquad
\widehat\M_{\mathbb Z}=\pi_{\widehat\omega_{\mathbb Z}}(\widehat\A_{\mathbb Z})''.
\]
Let \(s_{\mathbb Z}\) be the support projection of the normal extension of \(\widehat\omega_{\mathbb Z}\). If \(\sR_{\mathrm{hom}}\) is the homoclinic relation on \(\sX_{\mathbb Z}\), then there is a state-preserving isomorphism
\begin{equation}\label{eq:two-sided-support-corner}
s_{\mathbb Z}\widehat\M_{\mathbb Z}s_{\mathbb Z}
\cong W^*(\sR_{\mathrm{hom}},\mu_{\mathbb Z}).
\end{equation}
The projection \(s_{\mathbb Z}\) has central support one, and both algebras in \eqref{eq:two-sided-support-corner} and \(\widehat\M_{\mathbb Z}\) are factors. Their types are
\begin{equation}\label{eq:two-sided-factor-types}
s_{\mathbb Z}\widehat\M_{\mathbb Z}s_{\mathbb Z}\text{ has type }
\begin{cases}
\mathrm{III}_\lambda,&\kappa\ne1,\\
\mathrm{II}_1,&\kappa=1,
\end{cases}
\qquad
\widehat\M_{\mathbb Z}\text{ has type }
\begin{cases}
\mathrm{III}_\lambda,&\kappa\ne1,\\
\mathrm{II}_\infty,&\kappa=1.
\end{cases}
\end{equation}
\end{prop}

\begin{proof}
Let \(\widehat\sR_{\mathbb Z}\) be the finite-change relation on \(\sY_{\mathbb Z}\): two configurations are equivalent when they agree outside a finite set. The direct-integral construction in the proof of Proposition~\ref{prop:support-corner}, now using the regular representations of \(\widehat\sR_{\mathbb Z}\), shows that the GNS representation of \(\widehat\omega_{\mathbb Z}\) represents the saturation \(\widehat\sR_{\mathbb Z}[\sX_{\mathbb Z}]\). The decreasing projections imposing the hard-core rule on \([-n,n]\) converge strongly to the diagonal projection \(\one_{\sX_{\mathbb Z}}\). Compression by this projection restricts both endpoints of each finite-change arrow to \(\sX_{\mathbb Z}\), and hence gives \eqref{eq:two-sided-support-corner}. The compressed state is the canonical Feldman--Moore state and is faithful, so this projection is exactly \(s_{\mathbb Z}\). Since every point of the represented saturation is equivalent to a point of \(\sX_{\mathbb Z}\), the projection has central support one.

The relation \(\sR_{\mathrm{hom}}\) is hyperfinite, since it is the increasing union of the finite relations obtained by allowing changes only in \([-n,n]\). We next prove ergodicity. Let \(A\subset\sX_{\mathbb Z}\) be \(\sR_{\mathrm{hom}}\)-invariant and let
\[
\mathcal G_n^{\mathrm{out}}
=\sigma(x_i:|i|>n).
\]
After the exterior coordinates have been fixed, any two admissible fillings of \([-n,n]\) differ at only finitely many sites and hence lie in the same \(\sR_{\mathrm{hom}}\)-class. Every compatible filling has strictly positive conditional Gibbs probability. It follows that \(A\) is \(\mathcal G_n^{\mathrm{out}}\)-measurable modulo null sets for every \(n\).

Primitivity of \(P\) gives the two-sided exponential mixing estimate: there are constants \(C>0\) and \(0<\alpha<1\) such that, for \(m<n\),
\begin{equation}\label{eq:two-sided-mixing}
\bigl|\mu_{\mathbb Z}(B\cap C_0)
-\mu_{\mathbb Z}(B)\mu_{\mathbb Z}(C_0)\bigr|
\le C\alpha^{n-m}
\end{equation}
whenever \(B\in\sigma(x_i:|i|\le m)\) and \(C_0\in\mathcal G_n^{\mathrm{out}}\). One obtains \eqref{eq:two-sided-mixing} by conditioning on the two boundary variables and applying the Perron--Frobenius convergence of powers of \(P\) and of the time-reversed transition matrix. Given \(\varepsilon>0\), choose a central cylinder event \(B\) with \(\mu_{\mathbb Z}(A\mathbin\triangle B)<\varepsilon\). Since \(A\in\mathcal G_n^{\mathrm{out}}\) for every sufficiently large \(n\), letting \(n\to\infty\) in \eqref{eq:two-sided-mixing} yields
\[
\mu_{\mathbb Z}(A\cap B)
=\mu_{\mathbb Z}(A)\mu_{\mathbb Z}(B).
\]
As in the proof of Proposition~\ref{prop:hyperfinite-factor}, this implies
\[
\bigl|\mu_{\mathbb Z}(A)-\mu_{\mathbb Z}(A)^2\bigr|<2\varepsilon.
\]
Thus \(\mu_{\mathbb Z}(A)\in\{0,1\}\), and \(\sR_{\mathrm{hom}}\) is ergodic.

For \((x,y)\in\sR_{\mathrm{hom}}\), choose a finite interval containing all coordinates at which \(x\) and \(y\) differ. The DLR conditional weights on that interval are proportional to \(\kappa\) raised to the number of occupied sites. Therefore
\begin{equation}\label{eq:two-sided-cocycle}
D_{\mu_{\mathbb Z}}(x,y)
=\kappa^{\sum_{i\in\mathbb Z}(x_i-y_i)},
\end{equation}
where the sum is finite. Hence every cocycle value belongs to \(\kappa^{\mathbb Z}\).

It remains to prove that \(\kappa\) is an essential value. Let \(A\subset\sX_{\mathbb Z}\) have positive measure, and choose \(\delta>0\) such that
\[
p^2-(1+\kappa^{-1})\delta>0.
\]
By martingale convergence for the symmetric cylinder sigma-algebras, there are \(n\) and an admissible word \(u=(u_{-n},\ldots,u_n)\), with \(u_n=0\), such that, writing \([u]=\{x:(x_{-n},\ldots,x_n)=u\}\),
\[
\mu_{\mathbb Z}(A\cap[u])>(1-\delta)\mu_{\mathbb Z}([u]).
\]
The terminal zero can be required because every hard-core configuration has infinitely many zeros to the right. Define
\begin{align*}
D_u&=\{x\in[u]:(x_{n+1},x_{n+2})=(0,0)\},\\
R_u&=\{x\in[u]:(x_{n+1},x_{n+2})=(1,0)\},
\end{align*}
and let \(T_u:D_u\to R_u\) change only the coordinate \(n+1\) from \(0\) to \(1\). Then
\[
\mu_{\mathbb Z}(D_u)=p^2\mu_{\mathbb Z}([u]),
\qquad
D_{\mu_{\mathbb Z}}(T_u(x),x)=\frac{q}{p^2}=\kappa.
\]
For
\[
C=A\cap D_u\cap T_u^{-1}(A\cap R_u)
\]
we have
\begin{align*}
\mu_{\mathbb Z}(C)
&\ge \mu_{\mathbb Z}(D_u)
-\mu_{\mathbb Z}(D_u\setminus A)
-\mu_{\mathbb Z}\bigl(T_u^{-1}(R_u\setminus A)\bigr)\\
&\ge\bigl[p^2-(1+\kappa^{-1})\delta\bigr]
\mu_{\mathbb Z}([u])>0.
\end{align*}
Thus \(\kappa\) belongs to the ratio set. Since the relation is ergodic, its nonzero ratio set is a subgroup, and consequently
\[
r(\sR_{\mathrm{hom}},\mu_{\mathbb Z})\setminus\{0\}
=\kappa^{\mathbb Z}.
\]
The support-corner classification in \eqref{eq:two-sided-factor-types} now follows exactly as in Theorem~\ref{thm:type-classification}. When \(\kappa=1\), the cocycle in \eqref{eq:two-sided-cocycle} is identically one, so the canonical state is a trace; the ergodic hyperfinite relation has infinite classes and therefore gives the hyperfinite factor of type \(\mathrm{II}_1\). Indeed, a configuration with infinitely many occupied sites admits infinitely many single-site removals, while a configuration with only finitely many occupied sites has an eventually empty tail in which a particle can be inserted at infinitely many locations.

Because \(s_{\mathbb Z}\) is full, \(\widehat\M_{\mathbb Z}\) is a factor and has the same type \(\mathrm{III}_\lambda\) when \(\kappa\ne1\). Suppose \(\kappa=1\). The full corner is type \(\mathrm{II}_1\), so the ambient factor is semifinite. It cannot be type I, because every nonzero corner of a type I factor is type I. Let \(Q_n^{\mathbb Z}\) be the local projection onto configurations satisfying the hard-core rule on \([-n,n]\). Then \(s_{\mathbb Z}\le\pi_{\widehat\omega_{\mathbb Z}}(Q_n^{\mathbb Z})\) and \(\widehat\omega_{\mathbb Z}(Q_n^{\mathbb Z})=1\). If the ambient factor were finite, its normalized trace would restrict, through the faithful representation of the simple UHF algebra, to the unique UHF trace. Hence
\[
\tau\bigl(\pi_{\widehat\omega_{\mathbb Z}}(Q_n^{\mathbb Z})\bigr)
=\frac{F_{2n+3}}{2^{2n+1}}\longrightarrow0,
\]
contradicting
\[
0<\tau(s_{\mathbb Z})
\le\tau\bigl(\pi_{\widehat\omega_{\mathbb Z}}(Q_n^{\mathbb Z})\bigr)\longrightarrow0.
\]
Thus \(\widehat\M_{\mathbb Z}\) is type \(\mathrm{II}_\infty\).
\end{proof}

Propositions~\ref{prop:support-corner} and \ref{prop:two-sided-UHF} answer the relation between the two constructions. The path algebra does not discard a physically relevant sector: it extracts the support on which the hard-core Gibbs state becomes faithful. The distinction between type \(\mathrm{II}_1\) for the support corner and type \(\mathrm{II}_\infty\) for the full quasi-local GNS algebra at \(\kappa=1\) is a consequence of the infinitely many zero-probability directions retained by the full spin algebra.

\section{The centralizer and the flow of weights}\label{sec:centralizer}

Let
\[
\M^\omega=\{x\in\M:\sigma_t^\omega(x)=x\text{ for every }t\in\mathbb R\}
\]
be the centralizer of the faithful normal state \(\omega\).

\begin{lem}\label{lem:general-cylinder-weight}
For \(\xi=(\xi_0,\ldots,\xi_n)\in W_n\), let \(m(\xi)\) be the number of \(01\)-transitions as in \eqref{eq:m-xi}. Then
\begin{equation}\label{eq:general-cylinder-weight}
\mu(\cyl{\xi})
=\pi_0p^{n+\delta_{\xi_0,1}+\delta_{\xi_n,1}}
\kappa^{m(\xi)+\delta_{\xi_0,1}}.
\end{equation}
Consequently, if \(\xi,\eta\in W_n\) have the same terminal symbol, then
\begin{equation}\label{eq:general-weight-ratio}
\frac{\mu(\cyl{\xi})}{\mu(\cyl{\eta})}
=p^{\delta_{\xi_0,1}-\delta_{\eta_0,1}}
\kappa^{m(\xi)-m(\eta)+\delta_{\xi_0,1}-\delta_{\eta_0,1}}.
\end{equation}
\end{lem}

\begin{proof}
Let \(N_{ab}(\xi)\) be the number of \(ab\)-transitions in \(\xi\). Since \(11\) is forbidden,
\[
N_{01}(\xi)=m(\xi),
\qquad
N_{10}(\xi)=m(\xi)+\delta_{\xi_0,1}-\delta_{\xi_n,1},
\]
and hence
\[
N_{00}(\xi)=n-2m(\xi)-\delta_{\xi_0,1}+\delta_{\xi_n,1}.
\]
Using \(\pi_1=\pi_0q\) and \(q=\kappa p^2\), we obtain
\begin{align*}
\mu(\cyl{\xi})
&=\pi_{\xi_0}p^{N_{00}(\xi)}q^{N_{01}(\xi)}\\
&=\pi_0q^{\delta_{\xi_0,1}}
p^{n-2m(\xi)-\delta_{\xi_0,1}+\delta_{\xi_n,1}}
q^{m(\xi)},
\end{align*}
which is \eqref{eq:general-cylinder-weight}. The ratio formula follows.
\end{proof}

Define the kernel of the modular cocycle by
\begin{equation}\label{eq:Rmod}
\sR_{\mathrm{mod}}=\{(x,y)\in\sR: \ D_\mu(x,y)=1\}.
\end{equation}
For each \(n\), let
\begin{equation}\label{eq:Cn}
\C_n=
\operatorname{span}\bigl\{e^{(n)}_{\xi,\eta}:
\xi_n=\eta_n,\ \mu(\cyl{\xi})=\mu(\cyl{\eta})\bigr\}.
\end{equation}

\begin{prop}\label{prop:centralizer-description}
Each \(\C_n\) is a finite-dimensional \(C^*\)-subalgebra of \(\A_n\), the embeddings satisfy \(\iota_n(\C_n)\subset\C_{n+1}\), and
\begin{equation}\label{eq:centralizer-AF}
\M^\omega
=W^*(\sR_{\mathrm{mod}},\mu)
=\left(\bigcup_{n\ge0}\C_n\right)''.
\end{equation}
In particular, \(\M^\omega\) is a finite hyperfinite von Neumann algebra, and \(\omega|_{\M^\omega}\) is a faithful normal trace.
\end{prop}

\begin{proof}
By Proposition~\ref{prop:modular-matrix-units}, a local matrix unit is fixed by \(\sigma^\omega\) exactly when its two cylinder weights agree. Equal cylinder weight is an equivalence relation within each matrix block, so the corresponding matrix units form a finite direct sum of full matrix algebras; this proves that \(\C_n\) is a \(C^*\)-subalgebra. Appending the same admissible edge multiplies both cylinder weights by the same transition probability, which proves \(\iota_n(\C_n)\subset\C_{n+1}\).

Under the groupoid identification, \(e^{(n)}_{\xi,\eta}\) is supported on a bisection on which \(D_\mu\) is the constant \(\mu(\cyl{\xi})/\mu(\cyl{\eta})\). Every arrow in \(\sR_{\mathrm{mod}}\) belongs to one such bisection with constant value one. Hence the compact open bisections arising from the \(\C_n\)'s exhaust \(\sR_{\mathrm{mod}}\).

Let
\[
E_{\mathrm{mod}}:W^*(\sR,\mu)\longrightarrow
W^*(\sR_{\mathrm{mod}},\mu)
\]
be the canonical faithful normal conditional expectation associated with the measurable subequivalence relation \(\sR_{\mathrm{mod}}\subset\sR\). On bounded groupoid kernels it is restriction to \(\sR_{\mathrm{mod}}\). If \(e_{\mathrm{mod}}\) denotes the orthogonal projection of \(L^2(\sR,\nu_l)\) onto \(L^2(\sR_{\mathrm{mod}},\nu_l)\), then the Feldman--Moore construction gives
\begin{equation}\label{eq:subrelation-expectation-vector}
E_{\mathrm{mod}}(x)\Omega=e_{\mathrm{mod}}x\Omega,
\qquad x\in W^*(\sR,\mu).
\end{equation}
Now let \(x\in\M^\omega\). The vector \(x\Omega\) is fixed by the modular unitaries. In the Feldman--Moore model these unitaries act by multiplication with \(D_\mu^{it}\), so the spectral theorem implies that \(x\Omega\) is supported on \(\{D_\mu=1\}=\sR_{\mathrm{mod}}\). Thus \(e_{\mathrm{mod}}x\Omega=x\Omega\), and \eqref{eq:subrelation-expectation-vector} gives \(E_{\mathrm{mod}}(x)\Omega=x\Omega\). Since \(\Omega\) is separating for \(\M\), we obtain \(E_{\mathrm{mod}}(x)=x\), hence \(x\in W^*(\sR_{\mathrm{mod}},\mu)\). The reverse inclusion follows directly from \eqref{eq:FM-modular-general}. This proves \eqref{eq:centralizer-AF}.

The kernel relation is hyperfinite, so its von Neumann algebra is hyperfinite. Its cocycle is identically one, and therefore the canonical state is tracial. Since it is the restriction of the faithful normal state \(\omega\), it is a faithful normal tracial state; hence the centralizer is finite.
\end{proof}

\begin{rem}\label{rem:centralizer-blocks}
Lemma~\ref{lem:general-cylinder-weight} makes the blocks in \(\C_n\) explicit.
\begin{enumerate}
\item[\rm(i)] If \(\kappa=1\), two words of the same length and terminal symbol have equal weight exactly when they have the same initial symbol.
\item[\rm(ii)] If \(\kappa\ne1\) and \(p\notin\kappa^{\mathbb Z}\), equal weight is equivalent to equality of the initial symbols and equality of \(m(\xi)\).
\item[\rm(iii)] If \(p=\kappa^\ell\) for some \(\ell\in\mathbb Z\), equal weight is equivalent to
\[
m(\xi)+(\ell+1)\delta_{\xi_0,1}
=m(\eta)+(\ell+1)\delta_{\eta_0,1}.
\]
\end{enumerate}
Thus \(\C_n\) is obtained by grouping admissible words according to the indicated labels and taking one full matrix algebra on each nonempty group.
\end{rem}

\begin{thm}\label{thm:centralizer-factor}
Let \(R_{\mathrm{II}_1}\) denote the hyperfinite factor of type \(\mathrm{II}_1\). Then
\begin{equation}\label{eq:centralizer-structure}
\M^\omega\cong
\begin{cases}
R_{\mathrm{II}_1},&p\in\kappa^{\mathbb Z},\\[1mm]
R_{\mathrm{II}_1}\oplus R_{\mathrm{II}_1},&p\notin\kappa^{\mathbb Z}.
\end{cases}
\end{equation}
In particular, \(\M^\omega\) is a factor if and only if
\begin{equation}\label{eq:centralizer-factor-condition}
p\in\kappa^{\mathbb Z}.
\end{equation}
When \(p\notin\kappa^{\mathbb Z}\), the two minimal central projections are \(\one_{\sX_0}\) and \(\one_{\sX_1}\), and, with \(\tau_i\) the normalized trace on the corresponding summand,
\begin{equation}\label{eq:centralizer-state-decomposition}
\omega|_{\M^\omega}=\pi_0\tau_0\oplus\pi_1\tau_1.
\end{equation}
\end{thm}

\begin{proof}
By Proposition~\ref{prop:centralizer-description}, the center and the factor decomposition of \(\M^\omega\) are determined by the ergodic decomposition of \(\sR_{\mathrm{mod}}\). We first prove that the restrictions to \(\sX_0\) and \(\sX_1\) are ergodic for every value of the parameters.

Every \(x\in\sX_0\) has a unique decomposition into return blocks
\begin{equation}\label{eq:return-blocks}
x=(0,\widetilde\beta_1,\widetilde\beta_2,\ldots),
\qquad
\beta_j\in\{a,b\},
\qquad
\widetilde a=0,
\quad
\widetilde b=10.
\end{equation}
Equivalently, \(a\) records the return \(0\to0\), while \(b\) records the excursion \(0\to1\to0\). The resulting measurable bijection
\[
\Phi:\sX_0\longrightarrow\{a,b\}^{\mathbb N}
\]
carries \(\mu(\,\cdot\mid\sX_0)\) to the Bernoulli product measure
\begin{equation}\label{eq:return-Bernoulli}
\bigotimes_{j=1}^\infty(p\delta_a+q\delta_b).
\end{equation}
Indeed, each \(a\)-block contributes the transition weight \(p\), and each \(b\)-block contributes \(q\).

If \(u\) is an admissible prefix ending at \(0\), the prefix replacement
\begin{equation}\label{eq:block-transposition}
S_u:\cyl{u010}\longrightarrow\cyl{u100},
\qquad
S_u(u010z)=u100z,
\end{equation}
has derivative one, because both appended words have weight \(pq\). Under \(\Phi\), this is the adjacent transposition of an \(a\)-block and a \(b\)-block. Hence \(\sR_{\mathrm{mod}}|_{\sX_0}\) contains every finite permutation of the Bernoulli coordinates. By the Hewitt--Savage zero--one law \cite{HS}, every invariant measurable subset of \(\sX_0\) has conditional measure zero or one, so \(\sR_{\mathrm{mod}}|_{\sX_0}\) is ergodic.

Every \(x\in\sX_1\) begins with \(10\). The map
\begin{equation}\label{eq:X1-X0-map}
\theta(1,0,x_2,x_3,\ldots)=(0,x_2,x_3,\ldots)
\end{equation}
is a measure-space isomorphism from \((\sX_1,\mu(\,\cdot\mid\sX_1))\) to \((\sX_0,\mu(\,\cdot\mid\sX_0))\), and it carries \(\sR_{\mathrm{mod}}|_{\sX_1}\) to \(\sR_{\mathrm{mod}}|_{\sX_0}\). Thus the restriction to \(\sX_1\) is also ergodic.

Assume first that \(p\notin\kappa^{\mathbb Z}\). If \((x,y)\in\sR_{\mathrm{mod}}\), choose equal-length prefixes \(\xi,\eta\) ending at the same symbol and representing the arrow. Equation \eqref{eq:general-weight-ratio} gives
\[
1=p^{\delta_{\xi_0,1}-\delta_{\eta_0,1}}\kappa^m
\]
for some \(m\in\mathbb Z\). If \(\xi_0\ne\eta_0\), this would imply \(p\in\kappa^{\mathbb Z}\), contrary to the hypothesis. Hence every arrow preserves the initial symbol. The sets \(\sX_0\) and \(\sX_1\) are therefore precisely the two ergodic components of \(\sR_{\mathrm{mod}}\). Each restricted relation is hyperfinite, probability preserving after normalization, and has infinite classes because it contains all finite permutations of infinitely many Bernoulli return blocks. Its von Neumann algebra is consequently \(R_{\mathrm{II}_1}\). This proves the second line of \eqref{eq:centralizer-structure}, identifies the two central projections, and \eqref{eq:centralizer-state-decomposition} follows from \(\omega(\one_{\sX_i})=\mu(\sX_i)=\pi_i\).

Conversely, assume \(p=\kappa^\ell\) for some \(\ell\in\mathbb Z\). It remains only to connect the two ergodic components. If \(\ell\ge0\), set
\begin{equation}\label{eq:bridge-positive-ell}
u=(01)^{\ell+1}0,
\qquad
v=1\,0^{2\ell+2}.
\end{equation}
The words have the same length and end at \(0\); \(u\) starts at \(0\), \(v\) starts at \(1\), \(m(u)=\ell+1\), and \(m(v)=0\). Formula \eqref{eq:general-cylinder-weight} and \(p=\kappa^\ell\) give \(\mu(\cyl{u})=\mu(\cyl{v})\).

If \(\ell<0\), write \(\ell=-r\) with \(r\ge1\), and set
\begin{equation}\label{eq:bridge-negative-ell}
u=0^{2r},
\qquad
v=1(01)^{r-1}0.
\end{equation}
Again the words have the same length and terminal symbol, while \(m(u)=0\), \(m(v)=r-1\); equation \eqref{eq:general-cylinder-weight} and \(p=\kappa^{-r}\) give equal cylinder weights. In either case, the prefix replacement \(vz\mapsto uz\) is a positive-measure partial isomorphism in \(\sR_{\mathrm{mod}}\) from \(\sX_1\) to \(\sX_0\).

Let \(A\subset\sX\) be \(\sR_{\mathrm{mod}}\)-invariant. Ergodicity of the two restrictions shows that \(A\cap\sX_0\) and \(A\cap\sX_1\) are each null or conull in their respective components. The bridge excludes the two mixed possibilities, so \(A\) is null or conull. Thus \(\sR_{\mathrm{mod}}\) is ergodic. Proposition~\ref{prop:centralizer-description} now shows that \(\M^\omega\) is a finite hyperfinite factor. It is infinite-dimensional because its restriction to \(\sX_0\) contains the finite-permutation relation on infinitely many Bernoulli coordinates. Hence \(\M^\omega\cong R_{\mathrm{II}_1}\), proving the first line of \eqref{eq:centralizer-structure}.
\end{proof}

\begin{cor}\label{cor:flow-weights}
If \(\kappa\ne1\), put \(T=|\log\kappa|=|\log\lambda|\). The Connes--Takesaki flow of weights of \(\M\) is the periodic translation flow
\begin{equation}\label{eq:periodic-flow}
\mathbb R\curvearrowright\mathbb R/T\mathbb Z,
\qquad
s\cdot(t+T\mathbb Z)=t+s+T\mathbb Z.
\end{equation}
Equivalently,
\begin{equation}\label{eq:center-core}
Z\bigl(\M\rtimes_{\sigma^\omega}\mathbb R\bigr)
\cong L^\infty(\mathbb R/T\mathbb Z),
\end{equation}
with the dual action acting by translations.
\end{cor}

\begin{proof}
By Theorem~\ref{thm:type-classification}, \(\M\) is type \(\mathrm{III}_\lambda\) for \(\kappa\ne1\). The flow of weights of a type \(\mathrm{III}_\lambda\) factor, \(0<\lambda<1\), is the periodic translation flow with period \(|\log\lambda|\); see \cite{S,TakesakiII}.
\end{proof}

\begin{rem}\label{rem:semifinite-flow}
When \(\kappa=1\), \(\M\) is type \(\mathrm{II}_1\), so no type-\(\mathrm{III}\) flow invariant is needed. Under the continuous-core convention extended to semifinite factors, the modular action is exterior equivalent to the trivial action and
\[
\M\rtimes_{\sigma^\omega}\mathbb R
\cong\M\,\overline\otimes\,L(\mathbb R),
\qquad
Z\bigl(\M\rtimes_{\sigma^\omega}\mathbb R\bigr)
\cong L^\infty(\mathbb R),
\]
with the dual action equal to the ordinary translation flow on \(\mathbb R\). This is more precise than calling the semifinite flow ``trivial.''
\end{rem}

\section{Discussion and physical interpretation}\label{sec:discussion}

\subsection*{What changes at \texorpdfstring{\(\kappa=1\)}{kappa=1}}
The parameter \(\kappa\) has a direct local meaning: it is the Gibbs-weight ratio for inserting one particle into an empty three-site pattern while keeping the exterior configuration fixed,
\[
000\longmapsto010.
\]
At the level of the measured tail relation, the same operation is a partial isomorphism with Radon--Nikodym derivative \(\kappa\). The factor type is therefore controlled by the multiplicative obstruction accumulated under finite particle insertions and removals. For \(\kappa\ne1\), this obstruction has the discrete essential-value group \(\kappa^{\mathbb Z}\), and the factor is of type \(\mathrm{III}_\lambda\). At \(\kappa=1\), the cocycle is cohomologically trivial after the bounded change of density in \eqref{eq:coboundary}; an equivalent invariant probability measure appears and the support factor becomes finite.

This is a change in modular organization, not a loss of uniqueness or analyticity of the Gibbs state. The one-dimensional hard-core gas has a unique translation-invariant Gibbs state for every \(\kappa>0\) \cite{G,BR2}, and the Perron--Frobenius eigenvalue in \eqref{eq:PF-root} is simple and analytic throughout that interval. Thus no correlation length, pressure derivative, or local density becomes singular at \(\kappa=1\). The operator-algebraic transition records whether the finite-change cocycle can be removed by an equivalent change of measure.

\subsection*{Canonical purification and operational factor type}
Let \((\M,\M',L^2(\M,\omega),\Omega)\) be the standard GNS bipartite system. Because \(\M\) is a factor and \(\Omega\) is cyclic and separating, \(\M\) and \(\M'\) are commuting factors in Haag duality and jointly generate \(\B(L^2(\M,\omega))\). The vector \(\Omega\) therefore gives a canonical pure-state purification of the mixed Gibbs functional. In this setting, \cite[Theorem~52]{vLSWW25} characterizes type \(\mathrm{III}\) by arbitrary-precision LOCC convertibility between any two pure states and characterizes type \(\mathrm{III}_1\) by the corresponding statement without classical communication. The quantitative embezzlement invariant in \cite[Theorems~58 and~67]{vLSWW24} determines \(\lambda\) for type \(\mathrm{III}_\lambda\) factors. In the semifinite case, trace and majorization instead provide nontrivial pure-state entanglement monotones.

For the present model, this means that the standard-form purification has type-\(\mathrm{III}_\lambda\) operational behavior for \(\kappa\ne1\), with \(\lambda=\min\{\kappa,\kappa^{-1}\}\), and type-\(\mathrm{II}_1\) behavior at \(\kappa=1\). This interpretation must not be confused with a spatial bipartition of the thermal chain. Establishing an operational statement for two specified lattice regions would require identifying their represented local algebras and verifying the hypotheses used in the pure-state theory. The support-corner result of Section~\ref{sec:quasilocal} is nevertheless important here: it shows that the factor being classified is canonically contained in the GNS construction of the usual quasi-local Gibbs state, rather than arising from an unrelated auxiliary algebra.

\subsection*{A common path-space mechanism}
The comparison with Appendix~A of \cite{JNPWI25} can be made precise at the level of proof strategy. Both constructions start from a finite directed graph and an \(AF\)-algebra whose matrix units replace finite paths with a common terminal vertex. A distinguished state assigns a positive multiplicative weight to each path. The GNS von Neumann algebra is then analyzed through the tail equivalence relation, and its Connes invariant is generated by quotients
\[
\frac{w(\xi)}{w(\eta)}
\]
for finite paths \(\xi,\eta\) with the same endpoint. A density argument upgrades such local quotients from possible cocycle values to elements of the ratio set.

The input data are different. Here the graph is the two-vertex hard-core graph and the weights are stationary Markov weights; after restriction to paths beginning at \(0\), all endpoint-preserving quotients are powers of the single fugacity \(\kappa\). In \cite{JNPWI25}, the graph and weights encode fusion rules and quantum dimensions, and the subgroup generated by the corresponding local quotients may be discrete or dense in \(\mathbb R_{>0}\). Thus the common mechanism is a weighted path-space or Markov-cocycle computation, while the physical sources of the weights are respectively a classical exclusion gas and a topologically ordered quantum model.

\section*{Acknowledgments}
F.M. thanks Y.~Suhov for kind hospitality at DPMMS, University of Cambridge. Y.S. thanks IHES, Bures-sur-Yvette, for hospitality and support.

The authors express their gratitude to the anonymous reviewer for valuable insights that enhanced the clarity of this manuscript.

\section*{Declarations}
\noindent\textbf{Conflict of interest.}
The authors declare that they have no conflict of interest.

\smallskip
\noindent\textbf{Data availability.}
No data sets were generated or analyzed in this work.

\end{document}